\documentclass[11pt]{article}

\usepackage[bibliosources=refs.bib]{pomegranate}

\DeclareOperator{\per}
\DeclareOperator{\tr}
\DeclareOperator{\Exp}
\newcommand{\Herm}{\mathbb{H}}
\newcommand{\PD}{\Herm^{++}}
\newcommand{\PSD}{\Herm^{+}}
\newcommand{\CN}{\mathcal{CN}}
\newcommand{\eps}{\varepsilon}

\title{Optimal \(e^{(\gamma+o(1))n}\)-Approximation of the Permanent of Positive Semidefinite Matrices}
\author[1]{Nima Anari}
\author[2]{Farzam Ebrahimnejad}
\affil[1]{Stanford University, \texttt{anari@stanford.edu}}
\affil[2]{The Voleon Group, \texttt{f.ebrahimn@gmail.com}}
\date{}

\begin{document}

\maketitle

\begin{abstract}
We determine, up to lower-order terms in the exponent, the best possible deterministic polynomial-time approximation ratio for the permanent of a Hermitian positive semidefinite matrix. If \(A\succeq0\) has no zero diagonal entry, \(d=\rank(A)\), \(A=VV^\dagger\) with \(V\in\C^{n\times d}\) full column rank, and \(v_1,\ldots,v_n\) are the rows of \(V\), define
\[
  \Phi(V)=
  \max\set*{
    \sum_{i=1}^n \log(v_i^\dagger Xv_i)
    +\log\det X-\tr X+d
    \given
    X\succ0
  },
  \qquad
  \widehat P(A)=e^{\Phi(V)}.
\]
We prove the exact sandwich
\[
  e^{-\gamma n}\widehat P(A)\le \per(A)\le \widehat P(A).
\]
Here \(\gamma\) is the Euler--Mascheroni constant. Since the maximization is concave, this gives a deterministic polynomial-time \(e^{(\gamma+\eps)n}\)-approximation for every \(\eps>0\). Combined with the previous \(e^{(\gamma-\eps)n}\)-hardness of approximation for positive semidefinite permanents, this resolves the optimal exponential approximation ratio for deterministic polynomial-time algorithms as \(e^{(\gamma+o(1))n}\), assuming \(\Class{P}\ne\Class{NP}\). The proof is an entropy argument applied to the standard Wick integral formula for \(\per(A)\); the loss is exactly \(\gamma\) per factor because \(\E{\log T}=-\gamma\) for \(T\sim\Exp(1)\).

The result was obtained through interactions with GPT 5.5 Pro Extended: the first author's interaction was one-shot, while the second author's was a separate multi-turn interaction with high-level guidance. Both authors verified the theorem and proof. Codex was used to assemble and typeset the manuscript.
\end{abstract}

\section{Introduction}

The permanent of a matrix \(A\in\C^{n\times n}\) is
\[
  \per(A)=\sum_{\sigma\in S_n}\prod_{i=1}^n A_{i,\sigma(i)}.
\]
Computing it exactly is a central \(\Class{\#P}\)-hard problem \cite{Valiant1979}. For nonnegative matrices there is a celebrated randomized fully polynomial-time approximation scheme \cite{JerrumSinclairVigoda2004}, but the permanent of a positive semidefinite matrix is a different object. The entries may have signs or phases, while the Gram structure imposes strong algebraic constraints. In particular, the Wick representation in \cref{lem:wick} makes \(\per(A)\ge0\) transparent when \(A\succeq0\). It is natural to ask for the best possible exponential approximation factor achievable deterministically in polynomial time in this setting.

This paper proves that the optimal exponential constant is \(\gamma\), the Euler--Mascheroni constant, for deterministic algorithms up to arbitrarily small \(o(n)\) loss in the exponent. Together with the hardness theorem of Ebrahimnejad, Nagda, and Oveis Gharan \cite{EbrahimnejadNagdaOG2024}, this pins down the deterministic polynomial-time approximation ratio for positive semidefinite permanents at the exponential scale: assuming \(\Class{P}\ne\Class{NP}\), the optimal ratio is \(e^{(\gamma+o(1))n}\).

The algorithmic result follows from a short variational principle. Write \(A=VV^\dagger\), where \(V\in\C^{n\times d}\), \(d=\rank(A)\), and the rows of \(V\) are \(v_1,\ldots,v_n\in\C^d\). Consider
\begin{equation}\label{eq:phi-intro}
  \Phi(V)=
  \max\set*{
    \sum_{i=1}^n \log(v_i^\dagger Xv_i)
    +\log\det X-\tr X+d
    \given
    X\in\PD_d
  }.
\end{equation}
The maximization is concave. Its value is independent of the chosen full-rank Gram factor \(V\), since another factor differs from \(V\) by a unitary change of coordinates. The additive \(d\) is the Gaussian relative-entropy normalization; without it the bounds below would carry an extra \(e^d\) factor.

\begin{theorem}[Main theorem]\label{thm:main-intro}
Let \(A\succeq0\) be an \(n\times n\) Hermitian matrix, and let \(\widehat P(A)=\exp(\Phi(V))\), with the convention \(\widehat P(A)=0\) if \(A\) has a zero diagonal entry. Then
\[
  e^{-\gamma n}\widehat P(A)
  \le
  \per(A)
  \le
  \widehat P(A).
\]
Moreover, for every \(\eps>0\), \(\widehat P(A)\) can be approximated in deterministic polynomial time well enough to return an \(e^{(\gamma+\eps)n}\)-approximation to \(\per(A)\).
\end{theorem}

The constant \(\gamma\) enters for a simple probabilistic reason. If \(g\) is a standard complex Gaussian, then \(\abs{g}^2\) is an exponential random variable of mean \(1\), and
\[
  \E{\log \abs{g}^2}=-\gamma.
\]
The permanent of a positive semidefinite matrix has a Wick integral representation as a Gaussian moment. This representation is standard; in the PSD permanent literature it appears, for example, through Gurvits's quantum permanent framework and in the work of Anari, Gurvits, Oveis Gharan, and Saberi \cite{Gurvits2003,AnariGurvitsOGSaberi2017}. For completeness we give the short Wick proof in \cref{lem:wick}. The new step is to combine this representation with the Gibbs variational formula while leaving the Hermitian second-moment matrix as the optimization variable. The lower bound in \cref{thm:main-intro} is obtained by evaluating the variational formula on a Gaussian distribution with covariance \(X\). The upper bound is obtained by applying Jensen's inequality to an arbitrary trial distribution and then using the maximum-entropy upper bound with the same Hermitian second moment. These two inequalities differ by exactly \(\gamma n\).

\paragraph{Relation to prior constants.}
Anari, Gurvits, Oveis Gharan, and Saberi \cite{AnariGurvitsOGSaberi2017} gave a deterministic \(e^{(\gamma+1)n}\)-approximation based on a natural convex relaxation. Yuan and Parrilo \cite{YuanParrilo2022} connected the problem to maximizing products of linear forms and gave related certifiable relaxations. Ebrahimnejad, Nagda, and Oveis Gharan \cite{EbrahimnejadNagdaOG2024} recently improved the approximation factor to \(e^{(0.9999+\gamma)n}\) and proved \(\Class{NP}\)-hardness below \(e^{(\gamma-\eps)n}\) for every \(\eps>0\). Thus the Gaussian representation, product-of-linear-forms viewpoint, and convex-optimization framework are all part of the existing landscape. The present relaxation removes the remaining \(e^n\)-type loss by adding the entropy correction \(\log\det X-\tr X+d\) directly to that objective, and therefore matches the hardness threshold up to lower-order terms in the exponent.

\paragraph{Provenance.}
The first and second authors each obtained the result by interacting with GPT 5.5 Pro Extended.

In the first author's interaction, the model was given prior work on positive semidefinite permanents and asked to find an \(e^{\gamma n}\)-approximation; no follow-up prompt was given by the first author before the model produced the theorem, relaxation, and proof. Independently and in parallel, the second author obtained the same result through a separate multi-turn interaction involving high-level guidance. Both authors verified the resulting mathematics. Codex was used to assemble and typeset the manuscript. See \cref{app:provenance} for the corresponding note.

\paragraph{Organization.}
\Cref{sec:related} discusses related work. \Cref{sec:prelim} recalls the permanent, complex Gaussians, entropy, and the Gibbs variational formula. \Cref{sec:relaxation} states the relaxation formally. \Cref{sec:proof} proves the approximation sandwich. \Cref{sec:algorithm} explains deterministic computation. \Cref{app:provenance} records provenance.

\section{Related Work}\label{sec:related}

The permanent is one of the canonical counting polynomials. Valiant \cite{Valiant1979} proved that exact computation is \(\#\)P-hard even for \(0/1\) matrices. For matrices with nonnegative entries, Jerrum, Sinclair, and Vigoda \cite{JerrumSinclairVigoda2004} gave an FPRAS using rapidly mixing Markov chains. These algorithms do not apply to arbitrary complex matrices, while the positive semidefinite setting has a different algebraic structure coming from its Gram representation.

The positive semidefinite permanent was studied algorithmically by Anari, Gurvits, Oveis Gharan, and Saberi \cite{AnariGurvitsOGSaberi2017}, who gave a deterministic \(c^n\)-approximation with \(c=e^{\gamma+1}\). Their analysis used Gaussian representations of PSD permanents, a viewpoint closely related to Gurvits's earlier quantum permanent framework \cite{Gurvits2003}. Yuan and Parrilo \cite{YuanParrilo2022} studied a closely related product-of-linear-forms relaxation over the complex sphere and derived approximation certificates for Hermitian positive semidefinite matrices.

Hardness results for positive semidefinite permanents developed more recently. Meiburg \cite{Meiburg2023} proved constant-factor \(\Class{NP}\)-hardness via connections to quantum state tomography. Ebrahimnejad, Nagda, and Oveis Gharan \cite{EbrahimnejadNagdaOG2024} established the first exponential hardness of approximation, showing \(\Class{NP}\)-hardness below \(e^{(\gamma-\eps)n}\), and gave an \(e^{(0.9999+\gamma)n}\)-approximation. Their result identifies \(\gamma\) as the right target constant. The contribution here is a direct variational relaxation attaining that constant on the deterministic algorithmic side, up to the usual finite-precision loss. In combination, these results determine the optimal deterministic approximation ratio as \(e^{(\gamma+o(1))n}\) at the exponential scale.

The proof uses standard ingredients from Gaussian analysis and information theory: the Wick--Isserlis formula for Gaussian moments \cite{Isserlis1918,Janson1997}, the Gibbs variational principle, Jensen's inequality, and the Gaussian maximum-entropy principle \cite{CoverThomas2006}. We recall these facts in a self-contained way in \cref{sec:prelim}; they are not meant as new contributions. The novelty is the entropy-corrected variational relaxation and the matching \(\gamma n\) analysis obtained by applying these standard ingredients to the PSD permanent integral with the second-moment matrix left as an optimization variable.

\section{Preliminaries}\label{sec:prelim}

This section fixes notation and records standard facts used in the proof. We include the details to make the short argument self-contained and to make clear where the classical ingredients enter.

\paragraph{Matrices and Gram factors.}
All positive semidefinite matrices in this paper are Hermitian. For \(d\ge1\), let \(\PSD_d\) and \(\PD_d\) denote the cones of \(d\times d\) Hermitian positive semidefinite and positive definite matrices. We use the standard Gram representation: if \(A\succeq0\) has rank \(d\), then
\[
  A=VV^\dagger,
\]
where \(V\in\C^{n\times d}\) has full column rank. The \(i\)th row of \(V\) is denoted \(v_i^\dagger\), so that
\[
  A_{ij}=v_i^\dagger v_j.
\]
This is just the spectral decomposition of \(A\), with zero eigenvalues omitted. If \(A_{ii}=0\) for some \(i\), then \(v_i=0\). Equivalently, by Cauchy--Schwarz for the PSD inner product, \(\abs{A_{ij}}^2\le A_{ii}A_{jj}=0\) for all \(j\). Hence the \(i\)th row and column of \(A\) are zero, and \(\per(A)=0\). Thus the only interesting case is \(v_i\ne0\) for all \(i\).

\paragraph{Complex Gaussians.}
We use the standard circularly symmetric complex Gaussian convention. A random vector \(Z\sim\CN(0,I_d)\) has density
\[
  \pi^{-d}\exp(-z^\dagger z)
\]
with respect to Lebesgue measure on \(\C^d\cong\R^{2d}\). More generally, \(Z\sim\CN(0,X)\), for \(X\in\PD_d\), has density
\[
  \frac{1}{\pi^d\det X}\exp(-z^\dagger X^{-1}z).
\]
Its differential entropy is
\[
  h(Z)=d\log \pi+d+\log\det X.
\]
This is the usual real Gaussian entropy formula applied to the \(2d\)-dimensional real vector \((\Re Z,\Im Z)\); see, e.g., \cite[Chapter~8]{CoverThomas2006}. If \(g\sim\CN(0,1)\), then \(\abs{g}^2\sim\Exp(1)\), because the squared radius of a two-dimensional real Gaussian with covariance \(I_2/2\) is exponential. Consequently,
\[
  \E{\log \abs{g}^2}
  =
  \int_0^\infty e^{-t}\log t\,dt
  =
  \Gamma'(1)
  =
  -\gamma.
\]
We will use the shorthand
\begin{equation}\label{eq:euler-constant}
  \E{\log \abs{g}^2}=-\gamma.
\end{equation}

\paragraph{Entropy.}
For a probability density \(\mu\) on \(\C^d\), let
\[
  h(\mu)=-\int_{\C^d}\mu(z)\log\mu(z)\,dz
\]
denote differential entropy, whenever it is well-defined. We use the following standard real maximum-entropy theorem.

\begin{lemma}[Real Gaussian maximum entropy]\label{lem:real-max-entropy}
Let \(Y\) be an \(\R^k\)-valued random vector with density, finite second moment, and nonsingular covariance matrix
\[
  \Sigma=\E{(Y-\E{Y})(Y-\E{Y})^\top}.
\]
Then
\[
  h(Y)
  \le
  \frac12\log\det(2\pi e\,\Sigma)
  =
  \frac12\log\bracks*{(2\pi e)^k\det\Sigma}.
\]
Equality holds if and only if \(Y\) is distributed as \(N(\E{Y},\Sigma)\).
\end{lemma}

This is the usual Gaussian maximum-entropy principle; see, e.g., \cite[Chapter~8]{CoverThomas2006}. The form we need is the following consequence for complex random vectors when only the Hermitian second moment is retained.

\begin{lemma}[Complex entropy bound]\label{lem:complex-max-entropy}
Let \(\mu\) be a probability density on \(\C^d\) with finite second moment, let \(Z\sim\mu\), and set
\[
  M=\E_\mu{ZZ^\dagger}.
\]
Then \(M\succ0\), and
\begin{equation}\label{eq:max-entropy}
  h(\mu)\le d\log\pi+d+\log\det M,
\end{equation}
with equality at \(\CN(0,M)\).
\end{lemma}

\begin{proof}
If \(u^\dagger M u=0\), then \(u^\dagger Z=0\) almost surely. For \(u\ne0\), this contradicts absolute continuity with respect to Lebesgue measure on \(\C^d\). Thus \(M\succ0\).
Let \(m=\E_\mu{Z}\), write \(Z-m=R+iS\), and let \(\Sigma_{\R}\) be the real covariance matrix of \((R,S)\in\R^{2d}\). The Hermitian covariance
\[
  H=\E_\mu{(Z-m)(Z-m)^\dagger}
\]
also satisfies \(H\succ0\), by the same absolute-continuity argument. It does not by itself determine \(\Sigma_{\R}\); the missing information is the relation matrix
\[
  C=\E_\mu{(Z-m)(Z-m)^\top}.
\]
This extra matrix can only reduce the determinant relevant for entropy. Indeed, the augmented covariance of \((Z-m,\overline{Z-m})\) is
\[
  \Gamma=
  \begin{pmatrix}
    H & C\\
    \overline C & \overline H
  \end{pmatrix}
  \succeq0,
\]
and the linear change of variables from \((R,S)\) to \((Z-m,\overline{Z-m})\) gives \(\det\Gamma=4^d\det\Sigma_{\R}\). The Schur complement gives
\[
  \det\Gamma
  =
  \det H\cdot\det(\overline H-\overline C H^{-1}C)
  \le
  (\det H)^2.
\]
Therefore
\[
  \det\Sigma_{\R}\le 4^{-d}(\det H)^2.
\]
Since translation does not change differential entropy and the identification \(\C^d\cong\R^{2d}\) has unit Jacobian, \(h(\mu)=h(R,S)\). Applying \cref{lem:real-max-entropy} gives
\[
  h(\mu)\le d\log\pi+d+\log\det H.
\]
Finally, \(H=M-mm^\dagger\preceq M\), and determinant monotonicity on the PSD cone gives \(\det H\le\det M\). This proves \cref{eq:max-entropy}. Equality is achieved by \(Z\sim\CN(0,M)\), for which \(m=0\) and \(C=0\).
\end{proof}

\begin{lemma}[Gibbs variational formula]\label{lem:gibbs}
Let \(U:\C^d\to\R\cup\{-\infty\}\) be measurable and assume
\[
  Z_U=\pi^{-d}\int_{\C^d}\exp(U(z)-z^\dagger z)\,dz
\]
is finite and positive. Let \(\mathcal A_U\) be the set of probability densities \(\mu\) on \(\C^d\) for which \(\E_\mu{Z^\dagger Z}<\infty\), \(h(\mu)\) is well-defined, and \(\E_\mu{U(Z)}\) is well-defined, so that the expression below has no indeterminate \(\infty-\infty\) form. Then
\[
  \log Z_U
  =
  \sup\set*{
    \E_\mu{U(Z)}-\E_\mu{Z^\dagger Z}+h(\mu)-d\log\pi
    \given
    \mu\in\mathcal A_U
  },
\]
provided the Gibbs density proportional to \(\exp(U(z)-z^\dagger z)\) belongs to \(\mathcal A_U\).
\end{lemma}

\begin{proof}
This is the continuous form of the Gibbs variational principle, or equivalently the nonnegativity of relative entropy \cite[Chapter~2]{CoverThomas2006}.
Let \(\nu\) be the probability density proportional to
\[
  \pi^{-d}\exp(U(z)-z^\dagger z).
\]
For any \(\mu\in\mathcal A_U\), nonnegativity of relative entropy gives
\[
  0\le \D{\mu\river\nu}
  =
  \int \mu\log\mu
  -\E_\mu{U(Z)}+\E_\mu{Z^\dagger Z}+d\log\pi+\log Z_U.
\]
Rearranging yields the desired upper bound. Equality is attained by \(\mu=\nu\), which is admissible by assumption.
\end{proof}

\section{The Relaxation}\label{sec:relaxation}

Let \(A=VV^\dagger\succeq0\) have rank \(d\), with nonzero rows \(v_1,\ldots,v_n\). Define
\begin{equation}\label{eq:phi}
  \Phi(V)
  =
  \max\set*{\phi_V(X)\given X\in\PD_d},
  \qquad
  \phi_V(X)
  =
  \sum_{i=1}^n \log(v_i^\dagger Xv_i)
  +\log\det X-\tr X+d.
\end{equation}
The objective is concave on \(\PD_d\): each \(X\mapsto\log(v_i^\dagger Xv_i)\) is the logarithm of a positive linear functional, \(\log\det X\) is concave, and \(-\tr X\) is linear.
The product term is in the same family of product-of-linear-forms objectives explored in earlier relaxations for PSD permanents \cite{AnariGurvitsOGSaberi2017,YuanParrilo2022}. The additional term \(\log\det X-\tr X+d\) is the entropy correction used in this paper.
The \(+d\) is not an arbitrary shift: for complex Gaussians,
\[
  \log\det X-\tr X+d
  =
  -\D*{\CN(0,X)\river \CN(0,I_d)}.
\]
Equivalently, it is the dimension-dependent normalization coming from the Gaussian entropy formula. Without this term the bounds would carry an extraneous \(e^d\) factor rather than the exact sandwich in \cref{thm:sandwich}.

\begin{lemma}[Existence and invariance]\label{lem:exist-invariant}
Assume \(v_i\ne0\) for all \(i\) and that \(V\) has full column rank. Then the maximum in \cref{eq:phi} is finite and attained. Moreover \(\Phi(V)\) depends only on \(A=VV^\dagger\), not on the particular full-rank factor \(V\).
\end{lemma}

\begin{proof}
Since the vectors \(v_i\) span \(\C^d\), the term \(\log\det X\) forces \(\phi_V(X)\to-\infty\) when \(X\) approaches the boundary of \(\PSD_d\) while \(\tr X\) remains bounded. Indeed, the product terms are then bounded above by \(\log(\norm{v_i}^2\tr X)\), whereas \(\log\det X\to-\infty\). On the other hand, as \(\tr X\to\infty\),
\[
  \sum_{i=1}^n\log(v_i^\dagger Xv_i)+\log\det X
  \le
  \sum_{i=1}^n\log(\norm{v_i}^2\tr X)+d\log(\tr X/d),
\]
so the linear term \(-\tr X\) dominates. Therefore the maximum is attained on a compact subset of \(\PD_d\).

If \(W\) is another full-rank factor with \(WW^\dagger=VV^\dagger\), then the correspondence between the row vectors of \(V\) and \(W\) preserves all inner products on a spanning set. It therefore extends to a unitary change of coordinates, so \(W=VU\) for some unitary \(U\). The change of variables \(X\mapsto U^\dagger XU\) preserves \(\log\det X\), \(\tr X\), and the quantities \(v_i^\dagger Xv_i\). Hence \(\Phi(W)=\Phi(V)\).
\end{proof}

We write \(\Phi(A)\) for this invariant value and \(\widehat P(A)=e^{\Phi(A)}\). If \(A\) has a zero diagonal entry, we set \(\widehat P(A)=0\).

\section{Approximation Guarantee}\label{sec:proof}

The proof rests on the following Wick representation. As discussed above, this Gaussian representation is standard and has appeared before in work on PSD permanents \cite{Gurvits2003,AnariGurvitsOGSaberi2017}. For completeness we give a direct proof from the Wick--Isserlis formula.

\begin{lemma}[Wick formula for PSD permanents]\label{lem:wick}
Let \(A=VV^\dagger\succeq0\), with rows \(v_1,\ldots,v_n\), and let \(Z\sim\CN(0,I_d)\). Then
\[
  \per(A)=
  \E*{\prod_{i=1}^n \abs{v_i^\dagger Z}^2}
  =
  \frac{1}{\pi^d}
  \int_{\C^d}
  \prod_{i=1}^n \abs{v_i^\dagger z}^2
  e^{-z^\dagger z}\,dz.
\]
\end{lemma}

\begin{proof}
Write
\[
  v_i^\dagger Z=\sum_{a=1}^d \overline{(v_i)_a} Z_a,
  \qquad
  \overline{v_i^\dagger Z}=\sum_{b=1}^d (v_i)_b \overline{Z_b}.
\]
Then
\[
  \prod_{i=1}^n\abs{v_i^\dagger Z}^2
  =
  \sum_{a_1,\ldots,a_n}
  \sum_{b_1,\ldots,b_n}
  \prod_{i=1}^n
    \overline{(v_i)_{a_i}}(v_i)_{b_i}
  \prod_{i=1}^n Z_{a_i}
  \prod_{i=1}^n \overline{Z_{b_i}}.
\]
By the Wick--Isserlis formula for Gaussian moments \cite{Isserlis1918,Janson1997}, only pairings between \(Z\)-coordinates and conjugate coordinates contribute. In this complex circularly symmetric setting the basic contraction is \(\E{Z_a\overline{Z_b}}=\delta_{ab}\), while \(\E{Z_aZ_b}=\E{\overline{Z_a}\,\overline{Z_b}}=0\). Therefore
\[
  \E*{
    \prod_{i=1}^n Z_{a_i}
    \prod_{i=1}^n \overline{Z_{b_i}}
  }
  =
  \sum_{\sigma\in S_n}\prod_{i=1}^n \delta_{a_i,b_{\sigma(i)}}.
\]
Substituting this into the expansion gives
\[
  \E*{\prod_{i=1}^n \abs{v_i^\dagger Z}^2}
  =
  \sum_{\sigma\in S_n}
  \sum_{a_1,\ldots,a_n}
  \prod_{i=1}^n
    \overline{(v_i)_{a_i}}(v_i)_{a_{\sigma^{-1}(i)}}.
\]
In the last product, reindexing the second factor by \(i=\sigma(j)\) yields
\[
  \prod_{i=1}^n
    \overline{(v_i)_{a_i}}(v_i)_{a_{\sigma^{-1}(i)}}
  =
  \prod_{j=1}^n
    \overline{(v_j)_{a_j}}(v_{\sigma(j)})_{a_j}.
\]
Thus the sums over the coordinates factor, and
\[
  \E*{\prod_{i=1}^n \abs{v_i^\dagger Z}^2}
  =
  \sum_{\sigma\in S_n}
  \prod_{i=1}^n
  \sum_{a=1}^d
    \overline{(v_i)_a}(v_{\sigma(i)})_a
  =
  \sum_{\sigma\in S_n}\prod_{i=1}^n v_i^\dagger v_{\sigma(i)}
  =
  \per(A).
\]
The integral form is the expectation written using the density of \(Z\).
\end{proof}

\begin{theorem}[Entropic sandwich]\label{thm:sandwich}
For every Hermitian \(A\succeq0\),
\[
  e^{-\gamma n}\widehat P(A)\le \per(A)\le \widehat P(A).
\]
\end{theorem}

\begin{proof}
If \(A\) has a zero diagonal entry, then both sides are zero by convention, so assume all rows \(v_i\) are nonzero. Let
\[
  U(z)=\sum_{i=1}^n \log \abs{v_i^\dagger z}^2.
\]
By \cref{lem:wick},
\[
  \log\per(A)
  =
  \log\bracks*{
  \pi^{-d}\int_{\C^d}\exp(U(z)-z^\dagger z)\,dz
  }.
\]
The integral is finite because the integrand is a polynomial times a Gaussian density. It is positive because the zero set of \(\prod_i v_i^\dagger z\) is a finite union of proper complex hyperplanes, hence has Lebesgue measure zero. The corresponding Gibbs density is proportional to a polynomial times the standard Gaussian density, so it has finite second moment, finite entropy, and finite \(\E{U(Z)}\). Thus \cref{lem:gibbs} applies.

For the upper bound, fix an admissible trial density \(\mu\in\mathcal A_U\), and let \(M=\E_\mu{ZZ^\dagger}\). Jensen's inequality gives
\[
  \E_\mu{\log \abs{v_i^\dagger Z}^2}
  \le
  \log \E_\mu{\abs{v_i^\dagger Z}^2}
  =
  \log(v_i^\dagger Mv_i).
\]
Also \(\E_\mu{Z^\dagger Z}=\tr M\), and the complex entropy bound \cref{lem:complex-max-entropy} gives
\[
  h(\mu)\le d\log\pi+d+\log\det M.
\]
Substituting in the Gibbs formula,
\[
  \E_\mu{U(Z)}-\E_\mu{Z^\dagger Z}+h(\mu)-d\log\pi
  \le
  \sum_{i=1}^n\log(v_i^\dagger Mv_i)
  +\log\det M-\tr M+d
  \le
  \Phi(A).
\]
Here \(M\succ0\) by \cref{lem:complex-max-entropy}, so the last expression is one of the feasible values in \cref{eq:phi}. Taking the supremum over \(\mu\in\mathcal A_U\) yields \(\log\per(A)\le\Phi(A)\), or \(\per(A)\le \widehat P(A)\).

For the lower bound, take \(Z\sim\CN(0,X)\) for an arbitrary \(X\in\PD_d\). Then
\[
  v_i^\dagger Z\sim\CN(0,v_i^\dagger Xv_i),
\]
so by \cref{eq:euler-constant},
\[
  \E{\log \abs{v_i^\dagger Z}^2}
  =
  \log(v_i^\dagger Xv_i)-\gamma.
\]
For this Gaussian trial density, \(h(Z)=d\log\pi+d+\log\det X\) and \(\E{Z^\dagger Z}=\tr X\). Therefore the Gibbs formula gives
\[
  \log\per(A)
  \ge
  \sum_{i=1}^n\log(v_i^\dagger Xv_i)
  -\gamma n+\log\det X-\tr X+d
  =
  \phi_V(X)-\gamma n.
\]
Maximizing over \(X\in\PD_d\) gives
\[
  \log\per(A)\ge \Phi(A)-\gamma n.
\]
This is the desired lower bound.
\end{proof}

\section{Deterministic Computation}\label{sec:algorithm}

The relaxation \cref{eq:phi} is a concave maximization over the positive definite cone. Equivalently, \(-\phi_V\) is a convex function. This places the algorithmic part in the same convex-optimization framework as earlier PSD-permanent relaxations \cite{AnariGurvitsOGSaberi2017,YuanParrilo2022}. The value \(\Phi(A)\) is accessible by standard deterministic convex optimization, for example by the ellipsoid method or interior-point methods \cite{GroetschelLovaszSchrijver1993,NesterovNemirovskii1994}.

For completeness, we record a simple a priori bound on the optimizer.

\begin{lemma}[Trace normalization at the optimum]\label{lem:trace}
Let \(X_\star\) maximize \(\phi_V\). Then
\[
  \tr X_\star=n+d.
\]
\end{lemma}

\begin{proof}
At the optimum the first-order condition is
\[
  \sum_{i=1}^n
  \frac{v_i v_i^\dagger}{v_i^\dagger X_\star v_i}
  +X_\star^{-1}
  =
  I.
\]
Taking the trace after multiplying by \(X_\star\) gives
\[
  \sum_{i=1}^n
  \frac{\tr(X_\star v_i v_i^\dagger)}{v_i^\dagger X_\star v_i}
  +\tr(I_d)
  =
  \tr X_\star.
\]
Each term in the sum is \(1\), so \(\tr X_\star=n+d\).
\end{proof}

Thus the optimizer lies in the compact convex set
\[
  \set{X\in\PSD_d:\tr X=n+d}.
\]
The objective has a logarithmic barrier at the boundary, so an additive approximation to \(\Phi(A)\) can be obtained by standard weak optimization over a bounded convex body \cite{GroetschelLovaszSchrijver1993}. For rational input, one may work with a sufficiently accurate numerical spectral factorization of \(A\) and a spectrally truncated feasible region. The usual ellipsoid or self-concordant-barrier machinery gives polynomial dependence on the input size and \(\log(1/\eps)\) in the standard finite-precision model \cite{GroetschelLovaszSchrijver1993,NesterovNemirovskii1994}, and standard perturbation bounds for spectral factorizations control the effect of replacing \(A\) by the computed Gram factor \cite{Higham2002}. This is the same finite-precision convention used in earlier convex-relaxation algorithms for PSD permanents \cite{AnariGurvitsOGSaberi2017}.

\begin{theorem}[Algorithmic form]\label{thm:algorithm}
For every \(\eps>0\), there is a deterministic algorithm which, given a rational Hermitian positive semidefinite matrix \(A\in\C^{n\times n}\), runs in time polynomial in the input size, \(n\), and \(\log(1/\eps)\), and outputs a number \(Q(A)\) such that
\[
  e^{-(\gamma+\eps)n}Q(A)\le \per(A)\le Q(A).
\]
\end{theorem}

\begin{proof}
If \(A\) has a zero diagonal entry, output \(0\). Otherwise compute a sufficiently accurate numerical full-rank Gram factor \(V\), and solve the concave program \cref{eq:phi} to additive accuracy at most \(\eps n/2\). Let \(\widetilde\Phi\) be the resulting value with
\[
  \Phi(A)-\eps n/2\le \widetilde\Phi\le \Phi(A).
\]
Output
\[
  Q(A)=\exp(\widetilde\Phi+\eps n/2).
\]
Then \(Q(A)\ge e^{\Phi(A)}=\widehat P(A)\ge\per(A)\). Also
\[
  Q(A)\le e^{\Phi(A)+\eps n/2}
  \le e^{(\gamma+\eps/2)n}\per(A),
\]
by \cref{thm:sandwich}. Replacing \(\eps/2\) by \(\eps\) in the requested accuracy gives the stated guarantee. Standard spectral factorization perturbation bounds \cite{Higham2002} allow the Gram factor and objective values to be computed to the needed additive accuracy within the finite-precision model described above.
\end{proof}

\begin{remark}
The theorem is stated in the conventional finite-precision form. In exact real arithmetic, one can set \(Q(A)=\widehat P(A)\) and obtains the exact \(e^{\gamma n}\)-sandwich of \cref{thm:sandwich}.
\end{remark}

\section{Discussion}

The relaxation \cref{eq:phi} can be viewed as an entropy-corrected product-of-linear-forms relaxation, building on the product-of-linear-forms viewpoint in prior work \cite{AnariGurvitsOGSaberi2017,YuanParrilo2022,EbrahimnejadNagdaOG2024}. If one keeps only the product term
\[
  \sum_i \log(v_i^\dagger Xv_i)
\]
under a trace normalization, the analysis must separately pay for radial fluctuations of the complex Gaussian. The correction \(\log\det X-\tr X+d\) accounts for these fluctuations inside the variational problem itself. This is why the proof does not incur the additional \(e^n\)-scale loss that appeared in earlier approximation analyses \cite{AnariGurvitsOGSaberi2017,EbrahimnejadNagdaOG2024}.

The proof also suggests a broader template. Whenever a counting quantity admits an integral representation as
\[
  \int e^{U(z)}\,d\gamma(z),
\]
one may upper bound its log-partition function by replacing an arbitrary distribution with the maximum-entropy distribution of matching low-order moments, and lower bound it by evaluating the same variational formula on a structured trial family. The approximation factor is then the worst-case Jensen gap for that trial family. For PSD permanents, the relevant one-dimensional gap is exactly \(\gamma\).

\appendix

\section{Provenance}\label{app:provenance}

The first and second authors each obtained the result by interacting with GPT 5.5 Pro Extended.

In the first author's interaction, the model was given prior work on positive semidefinite permanents and asked to find an \(e^{\gamma n}\)-approximation. The first author did not provide mathematical insights, hints, proposed lemmas, proof strategies, or follow-up prompts before the model produced the theorem, relaxation, and proof.

Independently and in parallel, the second author obtained the same result through a separate interaction with roughly thirty back-and-forth messages. In that interaction, GPT 5.5 Pro Extended first explored incremental improvements and attempted to reanalyze the earlier convex program from prior work before arriving at the new entropy-corrected convex program and the \(e^{(\gamma+o(1))n}\)-approximation. The second author provided high-level guidance during this interaction.

Both authors verified the mathematics. Codex was used to assemble and typeset the paper.

\PrintBibliography

\end{document}